\documentclass{llncs}

\usepackage[lncs,amsthm,autoqed]{pamas}
\usepackage{arydshln}

\usepackage{times}
\usepackage{paralist}
\usepackage{array,xspace,multirow,epic}
\usepackage{natbib}
\usepackage{color}
\usepackage[boxed]{algorithm}
\usepackage{algorithmic}
\usepackage{booktabs}  %% nice tables
\usepackage{verbatim,ifthen}
\usepackage{enumitem}
\usepackage{pifont}
\usepackage{ifthen}
\usepackage{calrsfs,mathrsfs}
\usepackage{enumitem}
\usepackage{subfigure}
\usepackage{nicefrac}
\usepackage{paralist}

\usetikzlibrary{arrows}
\usetikzlibrary{decorations.pathreplacing}
\usetikzlibrary{patterns}

\newcommand{\guar}{guar}
\newcommand{\aar}{aar}

\newcommand\eat[1]{}

\setenumerate[1]{label=\rm(\it{\roman{*}}\rm),ref=({\it\roman{*}}),leftmargin=*}

\newlength{\wordlength}

\newcommand{\midd}{\mathbin{:}}

\newcommand{\eqclass}[2][]{\ifthenelse{\equal{#1}{}}{[#2]}{[#2]_{\sim_{#1}}}}

\newcommand{\Pref}[1][]{
	\ifthenelse{\equal{#1}{}}{\mathrel \succsim}{\mathop{\succsim_{#1}}}
}                                          
\newcommand{\sPref}[1][]{                  
	\ifthenelse{\equal{#1}{}}{\mathrel P}{\mathop{P_{#1}}}
}                                          
\newcommand{\Indiff}[1][]{                 
	\ifthenelse{\equal{#1}{}}{\mathrel I}{\mathop{I_{#1}}}
}
\newcommand{\prefset}[1][]{\ifthenelse{\equal{#1}{}}{\mathcal{R}}{\mathcal{R}_{#1}}}

\definecolor{kentuckyblue}{RGB}{0, 93, 170}			% Go Big Blue!
\definecolor{green}{RGB}{0, 102, 0}					% Haris Green
\definecolor{frenchred}{RGB}{250,60,50}				% French Red
\definecolor{unswyellow}{RGB}{255,155,0}			% USNW Yellow

\title{Egalitarianism of Random Assignment Mechanisms
}

\author{%
	Haris Aziz\inst{1,2} \and Jiashu Chen\inst{2} \and Aris Filos-Ratsikas\inst{3} \and \\ Simon Mackenzie\inst{1,2} \and Nicholas Mattei\inst{1,2}}
\institute{%
NICTA, Sydney, Australia\\
\email{\{haris.aziz,simon.mackenzie,nicholas.mattei\}@nicta.com.au}
\and
UNSW, Sydney, Australia\\
\email{jiashu.chen@student.unsw.edu.au }
\and 
University of Oxford, Oxford, UK and Aarhus University, Aarhus, Denmark\\
\email{filosra@cs.au.dk}
	}

\begin{document}

\maketitle

\begin{abstract}
We consider the egalitarian welfare aspects of random assignment mechanisms when agents have unrestricted cardinal utilities over the objects. We give bounds on how well different random assignment mechanisms approximate the optimal egalitarian value and investigate the effect that different well-known properties like ordinality, envy-freeness, and truthfulness have on the achievable egalitarian value. Finally, we conduct detailed experiments analyzing the tradeoffs between efficiency with envy-freeness or truthfulness using two prominent random assignment mechanisms ---  random serial dictatorship and the probabilistic serial mechanism --- for different classes of utility functions and distributions. 
\end{abstract}

%\end{document}
\vspace{-1em}
\section{Introduction}

We explore the tradeoffs between fairness and efficiency for randomized mechanisms for the \emph{assignment problem}. Specifically, we consider settings where $n$ agents express preferences (cardinal or ordinal) over a set of $m$ indivisible objects.  The objective is to assign the objects to agents in a fair and mutually beneficial manner~\citep[see \eg][]{ABS13a,AGMW14a,BoMo01a,HyZe79a}.  
% removed: Gard73b
This general setting has a number of important and significant applications including the assignment of tasks to cores in cloud computing, kidneys to patients in organ exchanges, runways to airplanes in transportation, and students to seats in schools.  

While it is sometimes assumed that the agents' preferences over the objects can be expressed fully through \emph{ordinal rankings}, most of the classical literature on the assignment problem \citep{BoMo01a,HyZe79a,Zhou90a} assumes the existence of an underlying utility structure, where each agent assigns real values or \emph{cardinal valuations} to the different objects; these are von Neumann-Morgenstern utilities that specify the intensity of the agents' preferences. 
%In the latter case, we often say that agents have \emph{cardinal valuations} over the objects. 
Some classical papers (see e.g. \cite{BoMo01a}) focus on \emph{ordinal mechanisms}, i.e. mechanisms that operate solely on the rankings consistent with the cardinal valuations, but \emph{cardinal mechanisms}, where agents explicitly report their numerical values,  have also been considered in literature~\citep{HyZe79a,Vazi07a,Zhou90a}. 

A well-established criterion for fairness is the Rawlsian concept of maximizing the happiness of the least satisfied agent~\citep{Rawl71a}. Following the spirit of this idea, we quantify the fairness of an allocation in terms of its \emph{egalitarian value}: the minimum ratio of the value of objects assigned to an agent to his total valuation for all the objects. The \emph{optimal egalitarian value} for a valuation profile of all agents is the best egalitarian value achievable over all assignments. The optimal egalitarian value is equivalent to the maximum egalitarian welfare if each agent has a total utility of one for the set of all objects. 
%%% NICK: I have removed this sentence.
%\nick{I don't think we need this sentence..} The latter assumption is often referred to as the \emph{unit-sum} normalization and is standard in the resource allocation literature \cite{CKKK12a,CGG13b,RKFZ14a,GuCo10a}.  
%removed: ChKo10a,
%\nick{Does it have to be 1 or just equal?}. Haris: yes it must be one.  
The advantage of considering the optimal egalitarian value is that it does not change if agents scale their relative values for the objects. %(i.e., it is invariant to affine transforms of their cardinal valuations). %and hence is not susceptible to scaling.

The egalitarian value is not the only criterion for desirable allocation mechanisms. Allocation mechanisms may have other goals and requirements such as envy-freeness or truthfulness. Crucially, both these properties are incompatible with optimizing the egalitarian value except in very restricted domains~\cite{BoMo04a}.
%\nick{cite here or you think this is well enough known?}. 
Thus, it is natural to examine the \emph{tradeoffs} between optimizing the egalitarian value and achieving other desirable properties. 
%like the ones mentioned above.  
In some settings, such as kidney exchanges, the tradeoff between fairness and efficiency is of the utmost concern \cite{DPS14a}.
 %\nick{do we want to mention efficiency here?  thats just the reference but we never mention it}. happy with the reference
 
 Evaluating these tradeoffs also motivates the study of how established mechanisms with other desiderata perform in terms of the egalitarian value. For a given mechanism $J$, we examine the approximation ratio $\guar(J)$, which is the minimum ratio (among all valuation profiles) of the egalitarian value of an allocation returned by the mechanism to the optimal egalitarian value. Our work falls under the umbrella of \emph{approximate mechanism design without money}, a framework set by \citet{PrTe13a} for the study of how well mechanisms with certain properties approximate some objective function of the agents' inputs.

In this paper, we study randomized assignment mechanisms for which achieving ex ante fairness is easier compared to deterministic mechanisms. Thus, to evaluate the performance of the mechanisms, we compare their egalitarian value with the optimal egalitarian value achieved by any randomized allocation. 
Note that computing the allocation with the optimal egalitarian value is an NP-hard problem when we restrict ourselves to deterministic allocations~\citep{DeHi88a}. On the other hand, when we consider randomized allocations, the optimal egalitarian value can be computed in polynomial time via a linear program. 

We give extra consideration to two randomized assignment mechanisms --- \emph{random serial dictatorship (RSD)} and \emph{probabilistic serial (PS)}, which are probably the best-known and most-studied mechanisms in the random assignment literature.
In RSD,\footnote{The original definition of RSD is for $n$ agents and $n$ objects; the definition here is a straightforward adaptation for $n <m$.} a permutation over the agents is selected uniformly at random and each agent in the permutation picks the most preferred $m/n$ units of object that are not yet allocated~\citep{ABB13b,BoMo01a,Sven99a}. %\nick{add a cite to your \#P paper or something here? thanks! done!}.
%removed: AzMe14b 
In PS, each object is considered to have an infinitely divisible probability weight of one.  To compute an allocation, agents simultaneously and with the same speed eat the probability weight of their most preferred object which has not been completely consumed. Once an object has been completely eaten by a subset of agents, each of these agents moves on to eat their next most preferred object that has not been completely eaten. The procedure terminates after all the objects have been eaten. The random allocation of an agent by PS is the amount of each object he has eaten~\citep[see \eg][]{BoMo01a,Koji09a}. PS satisfies stochastic dominance (SD) envy-freeness (envy-freeness with respect to all cardinal utilities consistent with the ordinal preferences).
We also formalize a mechanism which we refer to as \emph{Optimal Egalitarian and Envy-Free Mechanism (OEEF)}, which maximizes the egalitarian value of an allocation under the constraint that the allocation is envy-free.  Allocations under this mechanism can be computed in polynomial time via linear programming.

% \begin{itemize}[noitemsep]
% \item In RSD, a permutation over the agents is selected uniformly at random and each agent in the permutation picks the most preferred $m/n$ units of object that are not yet allocated~\citep{AzMe14b,BoMo01a,Sven99a}. %\nick{add a cite to your \#P paper or something here? thanks! done!}.
% \item In PS, each object is considered to have an infinitely divisible probability weight of one.  To compute an allocation, agents simultaneously and with the same speed eat the probability weight of their most preferred object which has not been completely consumed. Once an object has been completely eaten by a subset of agents, each of these agents moves on to eat their next most preferred object that has not been completely eaten. The procedure terminates after all the objects have been eaten. The random allocation of an agent by PS is the amount of each object he has eaten~\citep[see \eg][]{BoMo01a,Koji09a}.
% \end{itemize}

%\nick{We also provide a novel formalization of a mechanism we call the \emph{Optimal Egalitarian and Envy-Free Mechanism (OEEF)}.  This mechanism maximizes the egalitarian value of an allocation under the constraint that the allocation is envy-free.  Allocations under this mechanism can be computed in polynomial time via linear programming. Done! THANKS}

%\bigskip

\subsection{Our contributions}
We present novel theoretical and empirical results regarding fairness in randomized mechanisms.  Our main theoretical contributions are as follows.
 % : $n^{-1}\leq ar(PS) \leq 2n^{-1}$; $n^{-1}\leq ar(RSD) \leq 2n^{-1}$; $n^{-1}\leq ar(J) \leq n^{-1/5}$ for any envy-free mechanism, and $ar(J) \leq 2n^{-1}$ for any ordinal mechanism.

\begin{itemize}
	% \item For PS: $n^{-1}\leq \guar(PS) \leq 2n^{-1}$.
	% \item For RSD: $n^{-1}\leq \guar(RSD) \leq 2n^{-1}$.
		\item For \emph{any} SD envy-free mechanism $J$: $\guar(J) = O(n^{-1})$.
		\item For \emph{any} envy-free mechanism $J$: $\guar(J) = \Omega(n^{-1})$ and $ \guar(J) = O(n^{-1/5})$.
				\item For \emph{any} truthful-in-expectation mechanism $J$: $\guar(J) = O(n^{-1/5})$.
		\item For \emph{any} ordinal mechanism $J$: $\guar(J) = O(n^{-1})$.
%		\item \nick{No other mechanism satisfying envy-freeness has better $\guar$ than $\guar(OEEF)$?}
\end{itemize}
As a result of our general bounds, we also get asymptotically tight bounds of $\Theta(n^{-1})$ for RSD and PS.
As a result of our general bounds for envy-free mechanisms, we obtain bounds for well-known envy-free mechanisms such as \emph{competitive equilibrium with equal incomes (CEEI)}~\citep{Vazi07a} and the \emph{pseudo-market} mechanism~\citep{HyZe79a}. Since a random assignment of indivisible objects can also be interpreted as a fractional assignment of divisible objects, \emph{our results apply as well to fair allocation of divisible objects}. 

The constructions that provide the upper bounds for the $\guar$ values can be considered as extreme examples that may not be common in real-life scenarios. In order to better understand how the mechanisms may perform in practice, we consider the approximation ratio achieved by RSD and PS. 
We also examine the effect of imposing the envy-freeness constraint. 
We generate ordinal profiles via a Mallow's model for different levels of dispersion $\phi$ from a common reference ranking of objects, assigning cardinal utilities via the Borda and exponential scoring functions.  Sweeping $\phi$ from 0, where all agents have the same preference, to 1.0, where all preference orders are equally likely (the Impartial Culture), allows us to make statements regarding situations where agent preferences are more or less correlated.  We make the following observations.
\begin{itemize}
\item There is a negligible difference between the minimum and average achievable approximation ratios for PS and RSD under Borda utilities.  While PS preforms slightly better than RSD when agents have more extreme (exponential) utilities, both mechanisms preform strictly worse when agents' valuations are more similar, as they are under Borda utilities. 
\item When we require envy-freeness (as in OEEF) with exponential utilities, as $\phi$ increases towards 1.0 (i.e. Impartial Culture) the achievable approximation ratio first decreases slightly and then increases. Hence, as agents value more disparate objects highly, satisfying envy-freeness does not impose as stiff a penalty on the achievable approximation ratio.
\item In our experiments, the requirement of envy-freeness as a constraint in itself (as in the OEEF mechanism)  does not have a large impact on the OEV. However, since PS returns an SD envy-free (envy-free for all cardinal utilities consistent with the ordinal preferences) allocation, its achievable approximation ratio is strictly less than OEEF. 
\end{itemize}
\vspace{-1em}
\subsection{Related work}

The assignment problem has been in the center of attention in recent years in both computer science and economics
~\citep{Aziz14d,CKKK12a,GuCo10a,HSTZ11a,GZH+11a,GuNi12a,PPS12a}. %As we mentioned earlier, 
Often, in the classical assignment literature, agents are assumed to have an underlying cardinal utility preference structure, even if they are not asked to report it explicitly. On the other hand, there are many examples of well-known cardinal mechanisms, such the pseudo-market (PM) mechanism of \citet{HyZe79a} and the competitive equilibrium with equal incomes (CEEI) mechanism~\citep{Vazi07a}. Both mechanisms return allocations that are envy-free in expectation. The two prominent ordinal mechanisms in the literature are the probabilistic serial mechanism (PS) ~\citep{BoMo01a,ChKo10a,UKK+13a} and random serial dictator (RSD), a folklore mechanism that pre-existed the formulation of the assignment problem in \cite{HyZe79a}. Later, \citet{ChKo10a} proposed a variant of PS called \emph{multi-unit eating probabilistic serial (MPS)} that was formalised and axiomatically studied by \citet{Aziz14d}. 

The egalitarian welfare has received considerable interest within the computer science literature, especially for allocation of discrete objects in a deterministic manner. The problem is also referred to as the \emph{Santa Claus problem} in which the goal is to compute an assignment which maximizes the utility of the agent that gets the least utility~\citep[see \eg][]{AsSa10a,BeDa05a,FGM14a,NRR13a}.
For deterministic settings, \citet{DeHi88a} proved that the problem is NP-hard. On the other hand, for randomized/fractional allocations, the problem can be solved via a linear program.\footnote{Even a lexicographic refinement of the OEV maximizing allocations (in which the value of the worst off agent, then the second worst off agent, and so on, are maximized) can be computed in polynomial time via a series of linear programs~\citep{GKK10a}.}
%\nick{can you clarify the lexicographic refinement here?  It's not clear from context}
Recently, another fairness constraint that has been considered is the maxmin fair share~\citep{BoLe14a,PrWa14a}. The notion coincides with proportionality in the context of randomized/fractional allocations and hence is weaker than OEV. 

% In particular \citet{GuCo10a} and \citet{HSTZ11a} examined strategyproof allocation of divisible objects under linear utilities.
% \cite{CGG13a} proposed the PA mechanism that is strategyproof and approximated the Nash welfare utilities of the agents.

% \citet{GuCo10a}
%
% \citet{HSTZ11a}

% \citet{GZH+11a}
%
% \citep{GuNi12a,PPS12a}
%
% \citet{CKKK12a}

% \citet{PrTe13a}

Another popular objective is the maximization of the \emph{utilitarian welfare}, i.e. the sum of agents' valuations for an assignment.  \citet{RKFZ14a} proved that RSD guarantees $\Omega(n^{-1/2})$ of the total utilitarian welfare if the utilities are normalized to sum up to one for each agent, which is asymptotically optimal among all randomized truthful mechanisms.
In a recent paper, \citet{CFK+15a} proved similar results for the price of anarchy with respect to the utilitarian welfare of random assignment mechanisms, including RSD and PS. In this paper, we consider the effect on approximations of the egalitarian value from strategic aspects (truthful mechanisms), limited information (ordinal mechanisms), or additional fairness requirements (envy-free mechanisms). The egalitarian value does not require the agents' utilities to be normalized and does not require agents' utilities to be added. 
  
%\nick{these two are a little short... combine or beef them up. added a sentence and combined them}
\citet{BCK11a} determined the approximation ratio of RSD and PS when the objective is the maximization of a different notion, the ordinal social welfare, which is related to the ``popularity'' of an assignment~\citep{ABS13a,KMN11a}.
\citet{CKKK12a} examined the issue of how much efficiency loss fairness requirements like envy-freeness incur but crucially, their objective is maximization of utilitarian welfare.  %and not egalitarian welfare.

	\section{Preliminaries}

	An assignment problem is a triple $(N,O,v)$ such that $N=\{1,\ldots, n\}$ is the set of agents, $O=\{o_1,\ldots, o_m\}$ is the set of indivisible objects, and $v=(v_1,\ldots, v_n)$ is the valuation profile which specifies for each agent $i\in N$ utility or valuation function $v_{i}$ where $v_{ij}$ or $v_i(o_j)$ denotes the value of agent $i$ for object $o_j$. %We assume that that for each $j\in \{1,\ldots, m\}$,  $v_{ij}>0$ for some $i\in N$ and for each $i\in N$,  $v_{ij}>0$ for some $j\in \{1,\ldots,m\}$.\aris{Do we use these assumptions? Otherwise we can just remove them.}
We will denote by $V^n$ the set of all possible valuation profiles.

	A fractional or random allocation $p$ is a $(n\times m)$ matrix $[p(i)(j)]$ such that $p(i)(j) \in [0,1]$ for all $i\in N$, and $o_j\in O$. We denote by $\mathcal{P}$ the set of all feasible allocations.
	%, and  $\sum_{i\in N}x_{ij}= 1$ for all $o_j\in O$.
	% and  $\sum_{o_j\in O}x_{ij}= 1$ for all $i\in N$.
	The term $p(i)(j)$ which we will also write as $p(i)(o_j) $ represents the probability of object $o_j$ being allocated to  agent $i$. Each row $p(i) =(p(i)(1) ,\ldots, p(i)(m) )$ represents the \emph{allocation of agent $i$}. 
	The set of columns correspond to the objects $o_1,\ldots, o_m$.
	We will denote by $\hat{p}$ the $m$ dimensional vector where the $j$-th entry is $\sum_{i\in N}p(i)(j) $ is the total probability that object $j$ will be allocated to some agent.\footnote{We assume that objects can also be left unallocated with some probability.}
	The utility of agent $i$ from allocation $p$ is $u_i(p(i) )=\sum_{j\in O}(p(i)(j))v_{ij}.$ %The Nash welfare of an allocation $p$ is $\prod_{i\in N}v_i(p(i))$. \aris{Do we have any results about the Nash welfare? Should we just remove it?} 
	An allocation $p$ is \emph{proportional} if for all $i\in N$,  $u_i(p(i))\geq \frac{1}{n}u_i(O)$.
	An allocation $p$ is \emph{envy-free} if for all $i,j\in N$,  $u_i(p(i))\geq u_i(p(j))$. An allocation is \emph{SD envy-free} if it is envy-free with respect to all cardinal utilities consistent with the ordinal preferences.\footnote{SD envy-freeness also applies to cardinal mechanisms e.g., one that maximize total welfare subject to SD envy-freeness constraints.}
	
		We will consider randomized mechanisms that return a random allocation 
		for each instance of an assignment problem. Note the connection between random assignments for indivisible objects and fractional assignments of divisible objects; a random assignment can be viewed as a fractional assignment when agents have additive utilities over the objects. In that sense, we can use well-known mechanisms for fractional assignments, like the CEEI mechanism, as randomized mechanisms for our setting.

%	\nick{Randomized...random... do we want to use the word fractional or random? Haris: let's keep random. } 

We say that a mechanism is proportional if it always returns a proportional allocation. Similarly, a mechanism is envy-free if it always returns an envy-free allocation. A mechanism $M$ is \emph{truthful-in-expectation}, if for any agent $i \in N$, any valuation profile $v=(v_i,v_{-i})$ and any misreport $v_i'$ of agent $i$ it holds that $u_i(M(v_i,v_{-i})) \geq u_i(M(v_i',v_{-i}))$, i.e. no agent has any incentive to misreport her true valuation.

Two valuations $v_i$ and $v_i'$ are \emph{ordinally equivalent} if they induce the same ranking over objects, formally $v_{i}(o_j)\geq v_{i}(o_k)$ iff	$v_{i}'(o_j)\geq v_{i}'(o_k).$  A profile $v$ is ordinally equivalent to profile $v'$ if for each $i\in N$, $v_i$ and $v_i'$ are ordinally equivalent.  
A mechanism $J$ is \emph{ordinal} if for any two preference profiles $v$ and $v'$ that are ordinally equivalent, $J(v)=J(v')$, i.e., the allocations are the same for any pair of ordinally equivalent profiles. 
We now define the main efficiency measures that we will examine in the paper. 
% \aris{I changed solution concepts to efficiency measures because solutions concepts is usually reserved for equilibrium notions.}
\begin{itemize}
	\item The \emph{egalitarian value (EV)} of an allocation $p$ with respect to valuation profile $v$ is	$EV(p,v)=\inf\{\frac{u_i(p(i))}{u_i(O)} \midd i\in N\}.$
	
	\item For a given valuation profile, the \emph{optimal egalitarian value (OEV)} is the maximum possible egalitarian value that can be achieved %\nick{question: should these be \{'s below? FIXED THANKS}:
	$OEV(v)=\sup_{\lambda}\{\exists p\in \mathcal{P}\midd EV(p,v)=\lambda\}.$

\item For a given valuation profile $v$, an allocation $p$ \emph{achieves approximation ratio}
$\frac{EV(p,v)}{OEV(v)}.$

\item For a given mechanism $J$ and valuation profile $v$, we will say that $J$ \emph{achieves approximation ratio of $J$ for valuation $v$} as 
$\aar(J, v) = \frac{EV(J(v),v)}{OEV(v)}.$
	 
\item An allocation rule $J$ \emph{guarantees an approximation ratio} of $\guar(J)$ where $\guar(J)$ is defined as %\nick{note that $V^n$ is not technically defined here... you may want to add it above}:
$\guar(J)=\inf_{v\in V^n} \{\aar(J,v)\}.$
%\[\guar(J)=\inf_{v\in V^n} \frac{EV(J(v),v)}{OEV(v)}.\]

% an $\lambda$-approximation of OEV if for all valuations, the
% egalitarian value is at least $\lambda\in [0,1]$ of the optimal egalitarian value.

% \item A allocation property $P$ is compatible with an $\lambda$-approximation of  OEV if for all valuations, there exists an allocation satisfying $P$ with the egalitarian value is at least $\lambda\in [0,1]$ of OEV.
%
% \item An allocation rule property $P$ is compatible with an $\lambda$-approximation of OEV if there exists an allocation rule $f$ under which for all valuations, the
% egalitarian value under the rule $f$ is at least $\lambda\in [0,1]$ of OEV.
\end{itemize}
The guaranteed approximation ratio $\guar(J)$ is the worst-case guarantee over all instances of the problem that we will be looking to maximize in our theoretical results. % We conclude the section with an example of the achieved ratio of any anonymous mechanism on a specific valuation profile.
%
% \begin{example}
% 	Consider the setting with $n=3$ agents, $m=4$ objects, and each agent having the same Borda utilities over the objects: $4$ for the most preferred, $3$ for the second most preferred, $2$ for the third most preferred, and $1$ for the least preferred.  This gives each agent a total utility for the set of all objects of $10$. The optimal egalitarian welfare of $10/3$ is achieved when $1/3$ of each object is allocated to each agent.  This gives us an OEV equal to $(10/3)/10=1/3$.
% 	 $\frac{\nicefrac{10}{3}}{10} = \frac{1}{3}$.
% 	In the situation where all objects are equally valued by all the agents, any anonymous mechanism $J$ would give the same fraction of each object to each agent. Hence each agent gets $1/3$ of the value he has for the set of all objects. Thus $\aar(J, v) = \frac{EV(J(v),v)}{OEV(v)} = \frac{1/3}{1/3} = 1.$
% \end{example}

% \section{The Optimal Egalitarian and Envy-Free (OEEF) Mechanism}

%\nick{moved text from the EXP section: The additional envy-free constraint can easily be added to the linear program that maximized the egalitarian value. The mechanism that maximizes the egalitarian value while satisfying envy-freeness may be of independent interest and worthy of additional study. }

%\nick{add the LP here.}

\vspace{-1em}
\section{Theoretical Results}\label{sec:theory}

We first note that for a deterministic mechanism $J$, $\guar(J)=0$;  in the worst case, if all the agents only value the same object then all agents get zero utility except the agent who gets the valued object. 
% We also note that for some instances such as identical utilities for which OEV is only $n^{-1}$, even trivial mechanisms such as the uniform mechanism achieves an approximation ratio of $1$.
%
% 				\begin{observation}
% 		If agents have identical utilities, for any anonymous mechanism $J$, \[\frac{EV(J(v),v)}{OEV(v)}=1.\]
% 					\end{observation}
From now on, we will focus on randomized mechanisms. We start with the following lemma about proportional mechanisms.		

%
% \begin{lemma}
% 	For an allocation $p$, $EV(p,v)\geq \inf_{i\in n}((u_i(p(i))/(u_i(O)))$.
% 	\end{lemma}

		\begin{lemma}
			For any mechanism $J$  that is proportional, $\guar(J)\geq n^{-1}$.
			\end{lemma}
			\begin{proof}
				If the allocation $p$ is proportional then for each $i\in N$, $u_i(p(i))\geq n^{-1} (u_i(O))$. Since $EV(p,v)\geq  \inf_{i\in n}(u_i(p(i))/(u_i(O))$ and $(u_i(p(i))/(u_i(O)) \geq n^{-1}$ for all $i\in N$, we get that $EV(x)\geq n^{-1}$. Since $OEV(v)\leq 1$,  $\frac{EV(J(p,v))}{OEV(v)}\geq n^{-1}$.
			\end{proof}

		Since both PS and RSD are proportional~\citep[see \eg]{AzYe14a, BoMo01a}, we obtain the following guarantee on their approximation ratio.
		
		\begin{corollary}
			$\guar(PS)\geq n^{-1}$ and $\guar(RSD)\geq n^{-1}$.
%Both PS and RP guarantee $n^{-1}$ of the OEV. 
		\end{corollary}

	% \begin{corollary}
	% 	PS does not guarantee better than $n^{-1}$ of the OEV.
	% 	\end{corollary}

		\noindent A mechanism satisfies the \emph{favourite share} property if whenever all the agents have the same most preferred object then each agent is assigned to it with probability $n^{-1}$. We obtain the following theorem.
		%\nick{Would this be more clear as: A mechanism satisfies the \emph{favourite share} property if for any subset $S \subseteq N$ of agents all having the same most preferred object $o$, each agent in $S$ receives a $n^{-1}$ share of $o$}. %HARIS THIS def is not needed. 
				
				\begin{theorem}\label{thm:fav-share}
		For any mechanism $J$ that satisfies the favourite share property, $\guar(J) = O(n^{-1})$.% does  not guarantee better than $\frac{2}{n}$ of the OEV.
					\end{theorem}
					\begin{proof}
				Consider the following valuation profile with $n=m$, where $\epsilon$ is an arbitrarily small positive value.
						\[v_1(o_j) = \begin{cases} 1, &\mbox{if } j = 1 \\
					0,& \mbox{otherwise}. \end{cases}\]

					For $i\in \{2,\ldots, n\}$,

						\[v_i(o_j) = \begin{cases} 0.5+\epsilon, &\mbox{if } j = 1 \\
						 0.5-\epsilon,&\mbox{if } j = i\\
						  0,&\mbox{otherwise}.\end{cases}\]

\noindent Note that in the assignment which achieves $OEV$, agent $1$ is assigned object $o_1$ and the rest of the agents get utility $0.5-\epsilon$ and hence $OEV(v)=0.5-\epsilon$.  On the other hand, $J$ gives $1/n$ of $o_1$ to each of the agents so that agent $1$ gets utility $1/n$. Since $\epsilon$ can be arbitrarily small, it follows that $\guar(J)=O(n^{-1})$. 
						 \end{proof}
							 
							 \noindent We remark here that Theorem~\ref{thm:fav-share} holds even if agents have strict preferences; the utilities can be perturbed slightly to reflect strict preferences.  Theorem~\ref{thm:fav-share} gives us the following.

% 		\begin{theorem}
% For $n$ agents and $n+1$ objects and s, any mechanism that satisfies the favourite share property does  not guarantee better than $\frac{2}{n}$ of the OEV.
% 			\end{theorem}
% 			\begin{proof}
%
% 				\[v_1(o_j) = \begin{cases} 1 &\mbox{if } j = 1 \\
% 			0& \mbox{otherwise}. \end{cases}\]
%
% 			For $i\in \{2,\ldots, n\}$
%
% 				\[v_i(o_j) = \begin{cases} 0.5+\epsilon &\mbox{if } j = 1 \\
% 				 \frac{0.5-\epsilon}{n-1}& \mbox{otherwise}.\end{cases}\]
%
%
% 				$\begin{matrix}
% 					{\small Agent \backslash Object } & 1 & 2 & 3 & 4 & 5 & 6 & \cdots\\
% 					1 & 1.00 & 0.00 & 0.00 & 0.00 & 0.00 & 0.00 & \cdots\\
% 					2 & 0.5 + \epsilon & 0.5 - \epsilon & 0.00 & 0.00 & 0.00 & 0.00& \cdots\\
% 					3 & 0.5 + \epsilon & 0.00 & 0.5 - \epsilon & 0.00 & 0.00 & 0.00& \cdots\\
% 					4 & 0.5 + \epsilon & 0.00 & 0.00 & 0.5 - \epsilon & 0.00 & 0.00& \cdots\\
% 					5 & 0.5 + \epsilon & 0.00 & 0.00 & 0.00 & 0.5 - \epsilon & 0.00& \cdots\\
% 					6 & 0.5 + \epsilon & 0.00 & 0.00 & 0.00 & 0.00 & 0.5 - \epsilon& \cdots\\
% 				    \cdots & \cdots & \cdots & \cdots & \cdots & \cdots & \cdots & \cdots\\
% 				\end{matrix}$
%
% 				\begin{itemize}[noitemsep]
% 					\item Optimal Egalitarian Value $\rightarrow$ 0.5
% 					\item Egalitarian Value that the mechanism achieves is  $\frac{1}{n}$
% 				\end{itemize}
% 				\end{proof}

				\begin{corollary}
					$\guar(RSD)=O(n^{-1})$  and $\guar(PS)=O(n^{-1})$. % \aris{I removed the "for n agents and n objects" mentions. Although it is technically true, this case is a subcase of the model and hence an impossibility result for this case is sufficient. I think it looks more general this way.}
% 					\haris{Sure.}
%					The RP mechanism does not guarantee better than $\frac{2}{n}$ of the OEV for $n$ agents and objects.
					\end{corollary}

				\begin{theorem}\label{thm:sd-envyfree}
		For any mechanism $J$ that satisfies SD envy-freenes, $\guar(J) = O(n^{-1})$.% does  not guarantee better than $\frac{2}{n}$ of the OEV.
					\end{theorem}
					\begin{proof}
						SD envy-freeness implies the favourite share property. If an allocation does not satisfy the favourite share property, then the agent who gets less than $1/n$ of his most preferred object will be envious of another agent if he has extremely high utility for the object. 
						\end{proof}
		
	% \begin{theorem}
	% 	The RP mechanism does not guarantee better than $\frac{2}{n}$ of the OEV.
	% 	\end{theorem}
	% 	\begin{proof}
	% 		$\begin{matrix}
	% 			{\small Agent \backslash Object } & 1 & 2 & 3 & 4 & 5 & 6 & \cdots\\
	% 			1 & 1.00 & 0.00 & 0.00 & 0.00 & 0.00 & 0.00 & \cdots\\
	% 			2 & 0.5 + \epsilon & 0.5 - \epsilon & 0.00 & 0.00 & 0.00 & 0.00& \cdots\\
	% 			3 & 0.5 + \epsilon & 0.00 & 0.5 - \epsilon & 0.00 & 0.00 & 0.00& \cdots\\
	% 			4 & 0.5 + \epsilon & 0.00 & 0.00 & 0.5 - \epsilon & 0.00 & 0.00& \cdots\\
	% 			5 & 0.5 + \epsilon & 0.00 & 0.00 & 0.00 & 0.5 - \epsilon & 0.00& \cdots\\
	% 			6 & 0.5 + \epsilon & 0.00 & 0.00 & 0.00 & 0.00 & 0.5 - \epsilon& \cdots\\
	% 		    \cdots & \cdots & \cdots & \cdots & \cdots & \cdots & \cdots & \cdots\\
	% 		\end{matrix}$
	%
	% 		\begin{itemize}[noitemsep]
	% 			\item Optimal Egalitarian Value $\rightarrow$ 0.5
	% 			\item Egalitarian Value that RSD achieves ($EV_{RSD}$) $\rightarrow$ $\frac{1}{n}$
	% 		\end{itemize}
	% 			Ends the proof.
	% 		\end{proof}

	In the following, we will prove an upper bound on the approximation ratio of the OEV for two classes of mechanisms: envy-free mechanisms and truthful mechanisms. First we prove the following lemma that states that when looking for upper bounds on the approximation ratio, it suffices to only consider anonymous mechanisms. Similar lemmas have been proven before in literature \cite{RKFZ14a,GuCo10a}.
	
	\begin{lemma}\label{lem:anonym}
		Let $J$ be a mechanism with approximation ratio $\rho$. Then, there exists another mechanism $J'$ which is anonymous and has approximation ratio at least $\rho$. Furthermore, if $J$ is truthful or truthful-in-expectation, then $J'$ is truthful-in-expectation and if $J$ is envy-free, $J'$ is envy-free.
	\end{lemma}	
	
	\begin{proof}
		Let $J'$ be the mechanism that on input valuation profile $v$ first applies a uniformly random permutation to the set of agents and then runs mechanism $J$ on $v$. 
		Obviously, $J'$ is anonymous. Additionally, since $v$ can be an input to $J$ and the approximation ratio is calculated over all possible instances, the ratio of $J$ cannot be better than the ratio of $J'$. Finally, since the permutation is independent of valuations, if $J$ is truthful or truthful-in-expectation, $J'$ is truthful-in-expectation.
	\end{proof}	
	
	Now we state the following theorem, bounding the approximation ratio of any truthful-in-expectation mechanism. RSD and the uniform mechanism (that gives assignment probability of $1/n$ of each object to each agent) are strategyproof and ordinal mechanisms that both achieve a $\Theta(n^{-1})$ approximation of the OEV. We prove that for any truthful-in-expectation mechanism $J$, it holds that $guar(J)= O(n^{-1/5})$.
	
	\begin{theorem}\label{thm:truthful}
		For any truthful-in-expectation mechanism $J$,  $guar(J) = O(n^{-1/5})$.
	\end{theorem}
	
	\begin{proof}
		Let $J$ be a truthful-in-expectation mechanism; by Lemma \ref{lem:anonym}, we can assume without loss of generality that $J$ is anonymous. Consider the following valuation profile $v$ (summarized in Figure~\ref{figure:visual}) with $n = n_1 + {n_1}^2 + n_1^{5/2}$ agents and ${n_1}^2+1$ objects, where $\epsilon$ will be defined later:
		
		\begin{itemize}
			\item For every agent $i \in A=\{1,\ldots, n_1\}$, it holds that $v_i(1) = 1$ and $v_i(j)=0$ for every object $j \neq 1$.
			\item For every agent $i \in B=\{n_1+1, {n_1}^2\}$, it holds that $v_i(1)=1-\epsilon$, $v_i(i-n_1+1)=\epsilon$ and $v_i(j)=0$ for all objects $j \in O \backslash \{1,i-n_1+1\}$.
			\item For every $\ell = 1, \ldots, {n_1}^2$ and agent $i \in C_{\ell}=\{{n_1}^2+(\ell-1)\sqrt{n_1}+1, {n_1}^2+\ell \sqrt{n_1}\}$, it holds that $v_i(\ell+1)=1$ and $v_i(j) = 0$ for all objects $j \neq \ell$. 
		\end{itemize}

		\begin{figure}
		
				\vspace{-2em}
				\centering
		\scalebox{0.8}{
		\begin{tabular}{cl}

\small
		                $\mathbf{{n_1}^2 + 1}$ \bf{Objects} &
		                $\begin{matrix}
		                \bf{1} & \bf{2} & \bf{3}& \bf{4} & \cdots & \mathbf{n_1^2+1}\\
		                \end{matrix}$
		                \\
		                \\
		                $n_1$ agents, set $A$ &
		                $\begin{matrix}
		                        1 & 0 & 0 & 0 & \cdots & 0\\
		                        1 & 0 & 0 & 0 & \cdots & 0\\
		                        \vdots & \vdots & \vdots & \vdots & \ddots & \vdots\\
		                        1 & 0 & 0 & 0 & \cdots & 0\\
		                \end{matrix}$
		                \\
		                \\
		                ${n_1}^2$ agents, set $B$ &
		                $\begin{matrix}
		                        1-\epsilon & \epsilon & 0 & 0& \cdots & 0\\
		                        1-\epsilon & 0 & \epsilon & 0 & \cdots & 0\\
		                        1-\epsilon & 0 & 0.& \epsilon & \cdots & 0 \\
		                        \vdots  & \vdots & \vdots & \vdots & \ddots & \vdots \\
		                        1-\epsilon & 0& 0 & 0 & \cdots & \epsilon \\
		                \end{matrix}$
		                \\
		                \\
		                ${n_1}^{\frac{5}{2}}$ agents, set $C$ &
		                $\begin{matrix}
		                        0 & 1 & 0 & 0 & \cdots & 0\\
		                        \vdots & \vdots & \vdots & \vdots & \ddots & \vdots & \sqrt{n_1} \text{ agents, set } C_1\\
		                        0 & 1 & 0 & 0 & \cdots & 0\\
		                        0 & 0 & 1 & 0 & \cdots & 0\\
		                        \vdots & \vdots & \vdots & \vdots & \ddots & \vdots & \sqrt{n_1} \text{ agents, set } C_2 \\
		                        0 & 0 & 1 & 0 & \cdots & 0\\
		                        \vdots & \vdots & \vdots & \vdots & \ddots & \vdots & \vdots\\
		                                0 & 0 & 0 & 0 & \cdots & 1 \\
		                                \vdots & \vdots & \vdots & \vdots & \ddots & \vdots& \sqrt{n_1} \text{ agents, set } C_{n_1^2}\\
		                                0 & 0 & 0 & 0 & \cdots & 1\\
		                \end{matrix}$ \\

		                \end{tabular}
		                }
						\caption{Valuation profile $v$.}
							\vspace{-1.5em}
								\label{figure:visual}
						\end{figure}

		% \haris{Thanks for writing this formally in a nice way. Can we also supplement this a more visual display of valuations which makes it easier for the reader to digest the way the valuations are constructed?}\aris{Yes, I was thinking of the same thing. I will try to construct a visual instance, using the one that was already there for an earlier version of the proof. I hope we will have enough space to include the visualization.} 
		
		In other words, the instance consists of $n_1$ agents that have value $1$ for the first object (set $A$) and $0$ for everything else, ${n_1}^2$ agents that value object $1$ at $1-\epsilon$ and another object at value $\epsilon$ (set $B$) and $n_1^{5/2}$ agents that value a single object at $1$ (set $C= \cup_l C_l$), such that $\sqrt{n_1}$ agents value the object that some agent in the set $B$ has value $\epsilon$ for.
		
		Now if we let $\epsilon = 1/(n_1 - \sqrt{n_1})$, then the egalitarian value of the optimal allocation is at least $1/n_1$; an allocation with such a value is the following:
		\begin{itemize}
			\item Every agent $i \in A$ is allocated $1/n_1$ of object $1$.
			\item Every agent $i \in B$ is allocated $(n_1-\sqrt{n_1})/n_1)$ of the object they value at $\epsilon$.
			\item Every agent $i \in C_l$ is allocated $1/n_1$ of object $l+1$.
		\end{itemize}  
		Next consider a family of valuation profiles $\mathcal{V}$, consisting of profiles where all agents have the same valuations as in $v$, except one agent from $B$ that has value $1$ for the object that she had value $\epsilon$ in $v$ and $0$ for all other objects. Formally, for $\ell=1,\ldots,{n_1}^2$, we define a profile $v^{\ell} \in \mathcal{V}$ as follows:
		\begin{itemize}
			\item For every agent $i \neq n_1+\ell$, it holds that $v_i^{\ell}(j)= v_i(j)$ for all objects $j \in M$.
			\item For agent $n_1+\ell$, it holds that $v_{n_1+\ell}^\ell(\ell+1)=1$ and $v_{n_1+\ell}^\ell(j)=0$, for all objects $j \neq \ell+1$.
		\end{itemize}		
		Consider now any $\ell$ and the corresponding valuation profile $v^\ell$. Since $J$ is anonymous and agents in $C_{\ell} \cup \{n_1+\ell\}$ have identical valuations and since $|C_\ell|=\sqrt{n_1}$, the probability that agent $n_1+\ell$ is allocated object $\ell+1$ is at most $1/(\sqrt{n_1}+1)$ and her utility is hence at most $1/(\sqrt{n_1}+1)$. Now consider valuation profile $v$ and consider the probability $p(n_1+\ell)(\ell+1)$ that agent $n_1+\ell$ is allocated object $\ell+1$. By truthfulness, and since $v_{n_1+\ell}$ could be a misreport from $v_{n_1+\ell}^\ell$, it must hold that $p(n_1+\ell)(\ell+1) \leq 1/(\sqrt{n_1}+1) < 1/\sqrt{n}$. This implies that the contribution to the expected utility of agent $n_1+\ell$ from object $\ell+1$ is at most $\epsilon/\sqrt{n_1}$, which is at most $1/(n_1\sqrt{n_1}-n_1)$. 		
		
		Now consider the probability $p(n_1+\ell)(1)$ that agent $n_1+\ell$ is allocated object $1$. 
		From the arguments above, if $p(n_1+\ell)(1) < 1/(n_1\sqrt{n_1}-n_1)$, then the expected utility of agent $n_1+\ell$ is at most $2/(n_1\sqrt{n_1}-n_1)$ and the ratio is $O(1/\sqrt{n_1})$. Since $n=n_1+{n_1}^2+n_1^{5/2}$, that would mean that the theorem is proven. Hence, for $J$ to achieve a better ratio than $O(n^{-1/5})$, it has to be the case that for every $\ell=1,\ldots,{n_1}^2$, it holds that $p(n_1+\ell)(1) \geq 1/(n_1\sqrt{n_1}-n_1)$. This is not possible however, since then $\sum_{\ell=1}^{{n_1}^2}p(n_1+\ell)(1) \geq n_1/(\sqrt{n_1}-1) >1$. This completes the proof.
	\end{proof}

	Note that for utilitarian welfare maximization, Filos-Ratsikas et al. \cite{RKFZ14a} proved that an ordinal mechanism, RSD achieves the best approximation ratio among all truthful mechanisms. We conjecture that this is the case for the maximization of the egalitarian value as well, i.e. for any truthful mechanism $J$, $\guar(J)=O(n^{-1})$.
	
	We now turn our attention to envy-free mechanisms. For this class, we will prove an $O(n^{-1/5})$ upper bound as well; the proof actually uses the same valuation profile as the proof of Theorem \ref{thm:truthful}.
	
	\begin{theorem}\label{thm:ef}
		For any mechanism $J$ that satisfies envy-freeness, $\guar(J) = O(n^{-1/5})$.
	\end{theorem} 
	
	\begin{proof}
		Consider the valuation profile $v$ used in the proof of Theorem \ref{thm:truthful} and again consider the probability $p(n_1+\ell)(\ell+1)$ that agent $n_1+\ell$ is allocated object $\ell +1$.
		 Recall the definition of sets $A,B$ and $C$ from the proof of Theorem \ref{thm:truthful}. By envy-freeness, it holds that $p(n_1+\ell)(\ell+1) \leq 1/(\sqrt{n}+1) \leq 1/\sqrt{n}$ otherwise some agent $j \in C_{\ell}$ (who only values object $\ell+1$) would be envious of agent $n_1+\ell$.
		 
The rest of the steps are the same as in the proof of Theorem~\ref{thm:truthful}. Again, consider the probability $p(n_1+\ell)(1)$ that agent $n_1+\ell$ is allocated object $1$. 
Since $p(n_1+\ell)(1) < 1/(n_1\sqrt{n_1}-n_1)$, if $p(n_1+\ell)(1) < 1/(n_1\sqrt{n_1}-n_1)$ then for the same reasons mentioned in the last paragraph of the proof of Theorem \ref{thm:truthful}, we are done. Hence, we can assume that for every $\ell=1,\ldots,{n_1}^2$, it holds that $p(n_1+\ell)(1) \geq 1/(n_1\sqrt{n_1}-n_1)$. This is not possible however, since then $\sum_{\ell=1}^{{n_1}^2}p(n_1+\ell)(1) \geq n_1/(\sqrt{n_1}-1) >1$.  %\haris{can we direct the reader to the exact part of the proof where the argument is the same. I guess it is the last 2 paras? Still it will be better make the proof more obvious.  Although I wrote a previous version of the proof, it is not obvious even to  me where to link up the argument with the previous proof.}
	\end{proof}

\noindent From Theorem \ref{thm:ef}, we obtain the following corollary.
		
\begin{corollary}
	$\guar(CEEI) = O(n^{-1/5})$ and $\guar(PM) =O(n^{-1/5})$.
	\end{corollary}
	
	\noindent It would be interesting to provide a better bound for Theorem~\ref{thm:ef} or show it is optimal, i.e. come up with an envy-free mechanism that actually achieves the ratio. Finally, we consider the $OEV$ guarantees of ordinal mechanisms.
	%Next, we use a similar argument for mechanisms satisfying the favourite share property 
			 %Recently a strategyproof cardinal mechanism called the partial allocation mechanism has also been proposed. 
			 % We conjecture the following.
 %
 % 			\begin{conjecture}\label{conj:strategyproof}
 % 				For any strategyproof mechanism $J$, $\guar(J)\leq O(n^{-1})$.% doe
 % 				\end{conjecture}
				
				% \begin{lemma}
				% 	If any mechanism $J$, there exists an anonymous mechanism $J'$ such that $\guar(J')\geq \guar(J)$. If $J$ is strategyproof, then it follows that $J'$ is strategyproof.
				% 	\end{lemma}
				%
					% \begin{lemma}
					% 	Any anonymous and strategyproof mechanism satisfies the favorite share property.
					% 	\end{lemma}
					% 	\begin{proof}
					% 		Consider any valuation profile in which each agent has the same most preferred object. Then, each agent should get $1/n$ of the object. If some agent reports less and gets more of the object then the mechanism is not strategyproof. If some agents reports more than others and gets more then this means that the mechanism is not strategyproof again.
					% 		\end{proof}
					%

				   			 % We conjecture the following.
	 %
	 % 				    			\begin{conjecture}\label{conj:strategyproof}
	 % 				    				For any strategyproof mechanism $J$, $\guar(J)\leq O(n^{-1})$.% doe
	 % 				    				\end{conjecture}
					
					% \begin{corollary}
		% 				$guar(RSD) \leq n^{-1/5}$.
		% 				\end{corollary}
				
						\begin{theorem}
							For any mechanism $J$ that is ordinal, $\guar(J) =O(n^{-1})$.% does not guarantee better than $2n^{-1}$ of the OEV.
						\end{theorem}
						\begin{proof}
							Consider the setting with $n$ agents and $n+1$ objects $\{o^*, o_0, \ldots, o_{n-1}\}$. %where all agents have the same most preferred object, but a different second each agent has a unique second most preferred object and the preferences are symmetric. %\nick{I don't think the word symmetric has meaning here... it has not yet been defined.  I think I understand what you mean but it should be made more clear.}. 
							The preferences are as follows: each agent values $o^*$ the most. Agent $1$ has preference order 
							$o^*, o_0,\ldots, o_{n-2}, o_{n-1}$. The preference of each agent $i\in N\setminus \{1\}$ over the objects $O\setminus \{o^*\}$ are obtained as follows: take agent $i-1$ preference order over $O\setminus \{o^*\}$ and move the most preferred object of $i-1$ among $O\setminus \{o^*\}$ to the end of the preference order for agent $i$. 
							\vspace{-1em}
\begin{align*}
	\scriptsize
								1:&\quad o^*, o_0,\quad,\ldots, o_{n-2}, \quad\quad o_{n-1}\\
								2:&\quad o^*, o_1,\quad,\ldots, o_{n-1}, \quad\quad o_0\\ 
								% 3:&\quad o^*, o_2,\quad,\ldots, o_{n-2},\quad\quad o_1\\
								\vdots\\
								i:&\quad o^*, o_{i-1},\quad,\ldots, o_{n-i+1},\quad o_{i-2}
							\end{align*} 
							
							%Note that for any ordinal mechanism $J$, the ordinal preferences are symmetric \nick{Symmetric is not clearly defined.}.  
							By Lemma \ref{lem:anonym}, we can assume without loss of generality that $J$ is anonymous. Furthermore, since $J$ is ordinal, due to the preference profile, the mechanism cannot differentiate among the agents even though they may have different valuations over the objects.  
							Assume that there is some agent that is allocated at most $1/n$ of the universally most preferred object $o^*$. In this case case, consider the scenario where this agent has utility almost $1$ for $o^*$ and the other agents $i$ have utility $0.5+\epsilon$ for $o^*$ and utility $0.5-\epsilon$ for $o_{i-1}$ 
							where $\epsilon$ is an arbitrarily small positive value.. In this case, the egalitarian value achieved is $1/n$ whereas the OEV is almost $0.5$. Hence $\guar(J)=O(1/n)$.
						\end{proof}
						
						\noindent Since MPS is an ordinal mechanism, it follows that $\guar(MPS)= O(n^{-1})$. 
						% \begin{corollary}
% 							$\guar(MPS)= O(n^{-1})$.
% 							%\nick{note that MPS is only very briefly mentioned in the RR at this point... perhaps expanding there or removing here would be advisable.}
% 							% The PS mechanism does not guarantee better than $\frac{2}{n}$ of the OEV for $n$ agents and objects.
% 						\end{corollary}
	
\section{Experimental Results}
The results in Section~\ref{sec:theory} give us worst case bounds on the \emph{guaranteed approximation ratios ($\guar(J)$)} for a number of prominent randomized mechanisms including RSD and PS.  
%However, these results are derived from nuanced constructions and may not reflect the \emph{achieved approximation ratios ($\aar(J,v)$)} we see in practice.  
Hence, in this section we present experimental results which provide a perspective on what may happen ``in practice.'' Since PS can be considered as the most efficient SD envy-free mechanism (in view of various characterizations~\citep{BoHe12a,UKK+13a}), the results for PS can also be viewed as the effect of enforcing SD envy-freeness. 
In order to test the quality of RSD and PS we need to generate both preferences and cardinal utilities for the agents.  There are a number of generative statistical cultures that are commonly used to generate ordinal preferences over objects and the choice of model can have significant impact on the outcome of an experimental study (see e.g.~\cite{PRM13a}).%(see e.g.~\cite{PRM13a,RGMT06a}).  % The choice of preference model is important and has significant impact on the outcome of an experimental study (see e.g.~\cite{PRM13a,RGMT06a}). 

Since our focus is fairness, and fairness is often hard to achieve when agents have similar valuations, we employ the \emph{Mallows model} \cite{Mall57a} and use the generator from \textsc{www.preflib.org} \cite{MaWa13a} in our study.  Mallows models are often used in machine learning and preference handling as they allow us to easily control the correlation between the preferences of the agents; a common phenomena in preference data \cite{MaWa13a,Matt11a,LuBo11a}. A Mallows model has two parameters:
\begin{inparaenum}[(1)]
	\item a \textbf{Reference Order ($\sigma$)}, the preference order at the center of the distribution, and 
\item	a \textbf{Dispersion Parameter ($\phi$)},  the variance in the distribution which controls the level of similarity of the agent preference orders. When $\phi=0$ all agents have the same ordinal preference; when $\phi=1$ then the ordinal preferences are drawn uniformly at random from the space of all preference orders.
	\end{inparaenum}

% \noindent
% \textbf{Reference Order ($\sigma$)}: The preference order at the center of the distribution.\\
% \textbf{Dispersion Parameter ($\phi$)}:  The variance in the distribution which controls the level of similarity of the agent preference orders. When $\phi=0$ all agents have the same ordinal preference; when $\phi=1$ then the ordinal preferences are drawn uniformly at random from the space of all preference orders.\\

Formally, the probability of observing an ordering $r$ is inversely proportional to the Kendall Tau distance between $\sigma$ and $r$.  This probability is weighted by $\phi$, which allows us to control the shape of the distribution. %Intuitively, Mallows models allow for the imitation of a ``social norm'' or ``ground truth'' ranking of which all agents are a noise observation.  Agents preferences are more or less correlated depending on the value we give to $\phi$, indicating the strength of the social norm or the ``correctness'' of the ground truth.  
For a given ordinal preference, we superimpose cardinal utilities for the agents using two well-established scoring functions:
	\begin{inparaenum}[(1)]
		\item \textbf{Borda Utilities}, each agent has valuation of $m-i$ for his $i$-th preferred object, and 
		\item \textbf{Exponential Utilities}, each agent has valuation of $2^{m-i}$ for his $i$-th preferred object.
	\end{inparaenum}

	% \begin{description}[noitemsep]
	% 	\item [Borda Utility:] Each agent has valuation of $m-i$ for his $i$-th preferred object.
	% 	\item [Exponential Utility:] Each agent has valuation of $2^{m-i}$ for his $i$-th preferred object.
	% \end{description}
	
In our experiments we generate 10,000 valuation profiles (instances) for each combination of parameters with the number of agents $n \in \{2, \ldots, 9\}$, number of objects $m \in \{2, \ldots, 9\}$, and dispersion parameter $\phi \in \{0.0, 0.1, \ldots, 1.0\}$.  We draw $\sigma$ i.i.d. for each instance.
\vspace{-1em}
	\subsection{Experiments: The Performance of RSD and PS}
For each instance $v$ generated, and each mechanism $J$, we examined
the achieved approximation ratio, $\aar(J,v) = \frac{EV(J(v),v)}{OEV(v)}$, of the RSD and PS mechanisms. Among all such values computed, we examined the minimum and average ratio achieved for a given set of parameters.  The results of our experiments for Borda Utilities are shown in Figure~\ref{fig:min-avg-borda} while the results for exponential utilities are shown in Figure~\ref{fig:min-avg-exp}.  
All our figures are aggregations over all values of $m$ for particular combinations of $n$ and $\phi$.  Note that since $aar(J,v)$ is normalized over the total utility we can aggregate these terms as it is invariant to this scaling.  This allows us to draw more general conclusions as we range over the number of agents and objects.
Empirically we found that increasing the number of agents has a greater impact on the approximation performance of the mechanisms compared to an increase in the number of objects, hence the decision to aggregate the graphs in the manner chosen.  This empirical result is in line with our theoretical results showing that the worst case approximation ratio is a function of $n$. The results for both mechanisms in Figure~\ref{fig:min-avg-borda} strictly dominate the results in Figure~\ref{fig:min-avg-exp}.  Hence, we observe that the achieved approximation ratio is better for Borda utilities than for exponential utilities.  
%One possible reason for this behaviour, supported by our theoretical results, is that as the difference between the valuations of the object grows large, fairness is harder to achieve.

\begin{figure}[ht!]
		\minipage[t][][b]{0.45\textwidth}
		\centering
		 \footnotesize{RSD Min. Achieved Approx. Ratio \\ Borda Utilities}
		  \includegraphics[scale=0.20]{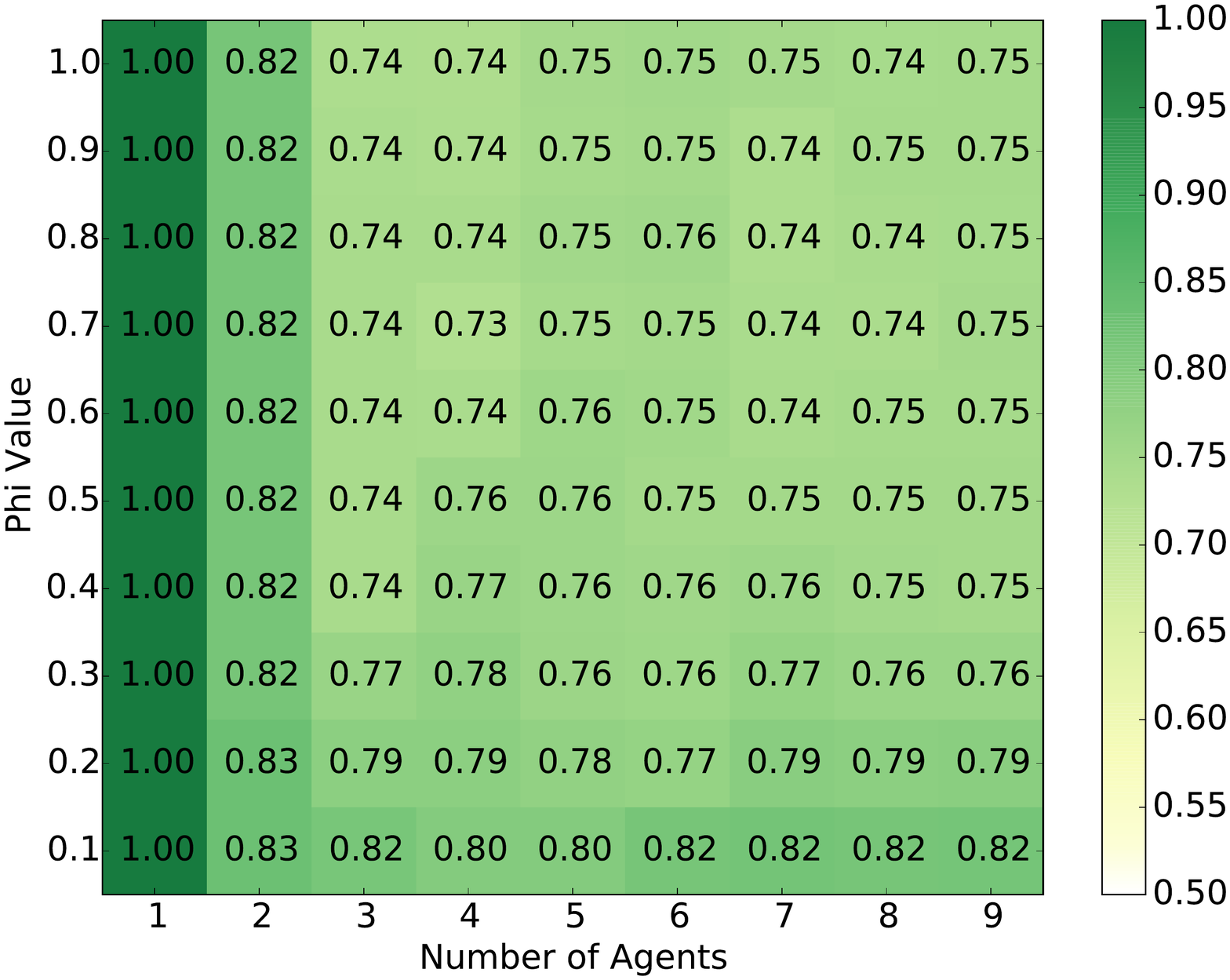}
		\endminipage
	\hfill
		\minipage[t][][b]{0.45\textwidth}
		\centering
		\footnotesize{PS Min. Achieved Approx. Ratio \\ Borda Utilities}
  		\includegraphics[scale=0.20]{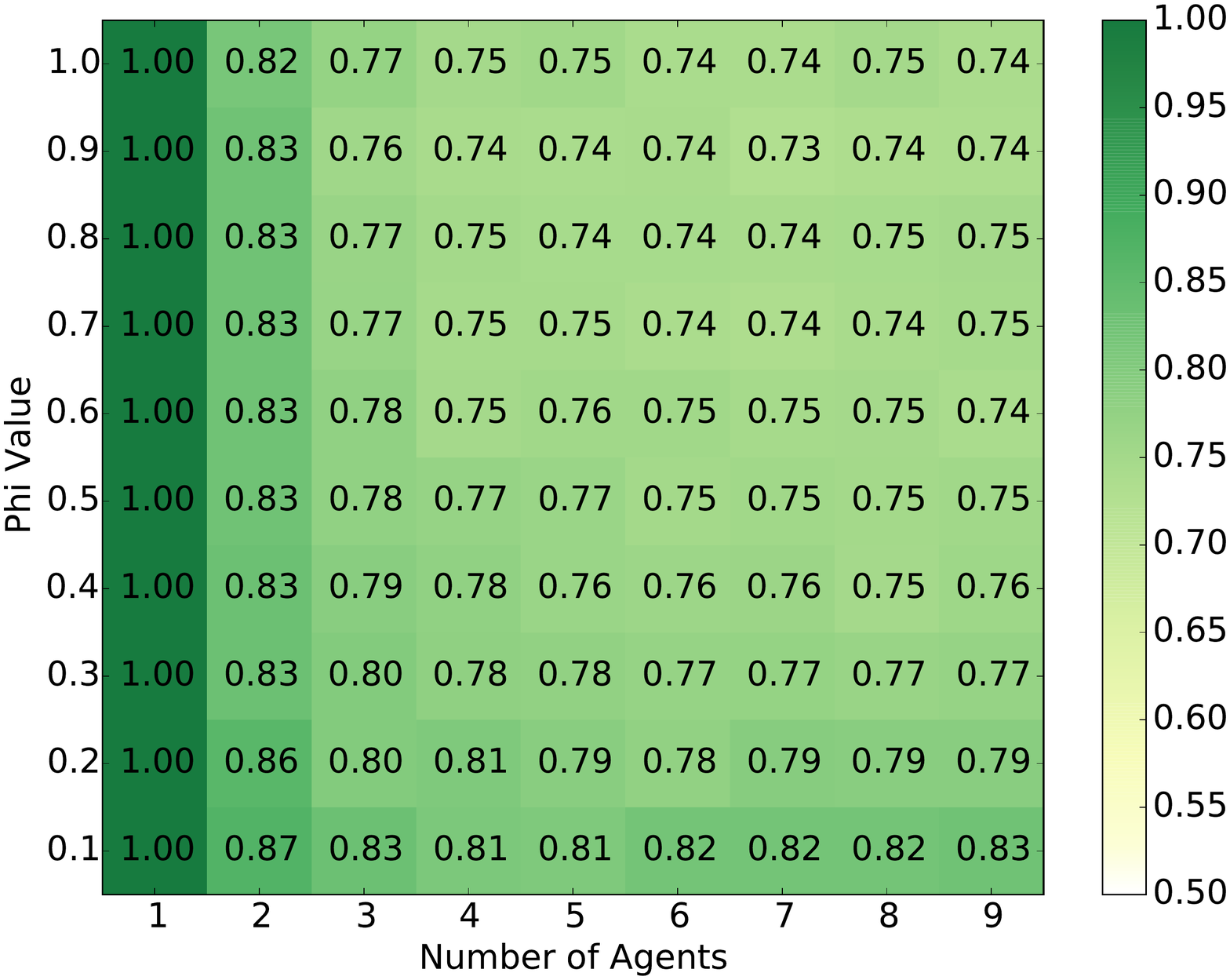}
		\endminipage
		
		\vspace{0.5cm}

		\minipage[t][][b]{0.45\textwidth}
		\centering
		\footnotesize{RSD Mean Achieved Approx. Ratio \\ Borda Utilities}
	 	\includegraphics[scale=0.20]{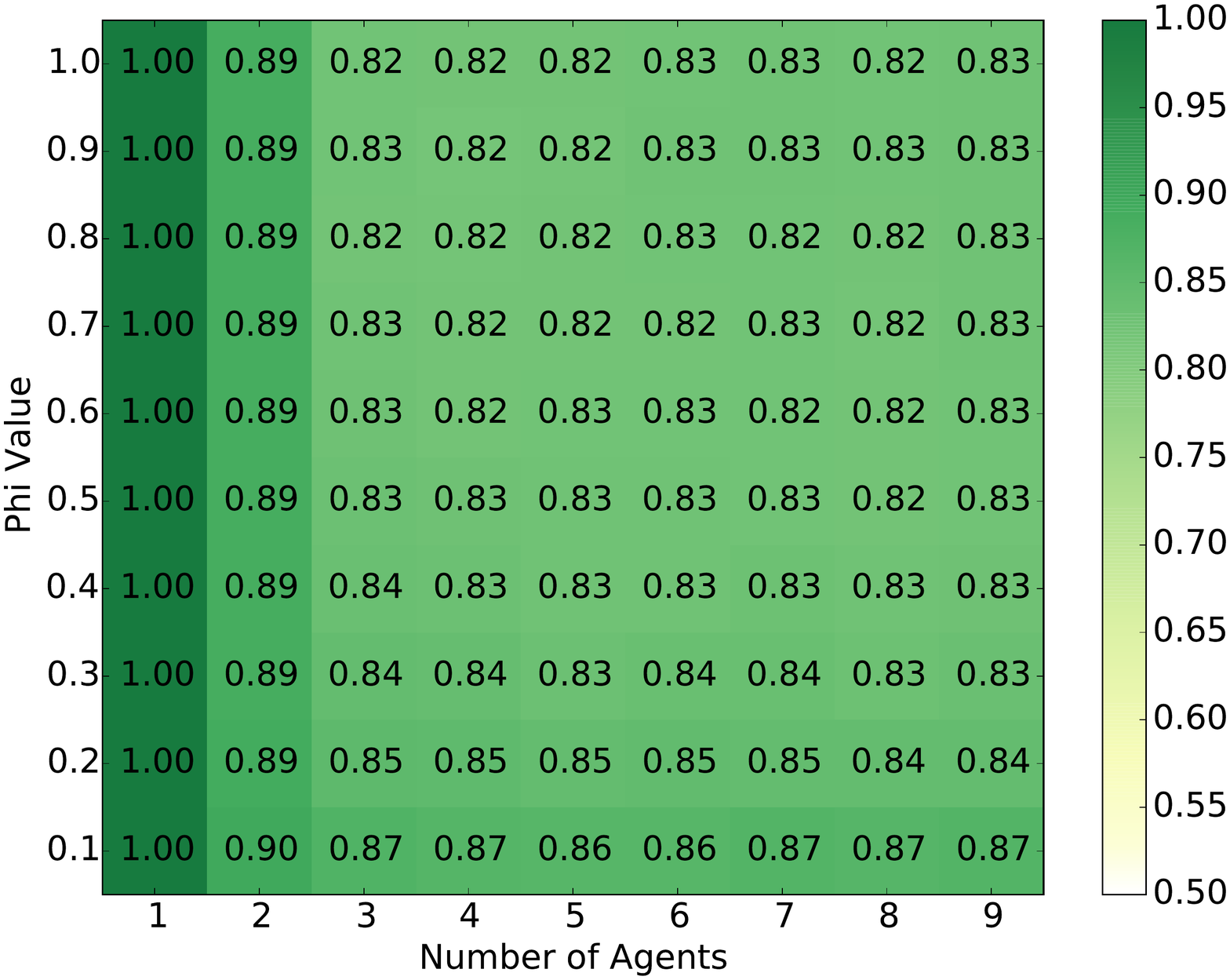}
		\endminipage
	\hfill
		\minipage[t][][b]{0.45\textwidth}
		\centering
		\footnotesize{PS Mean Achieved Approx. Ratio \\ Borda Utilities}
  		\includegraphics[scale=0.20]{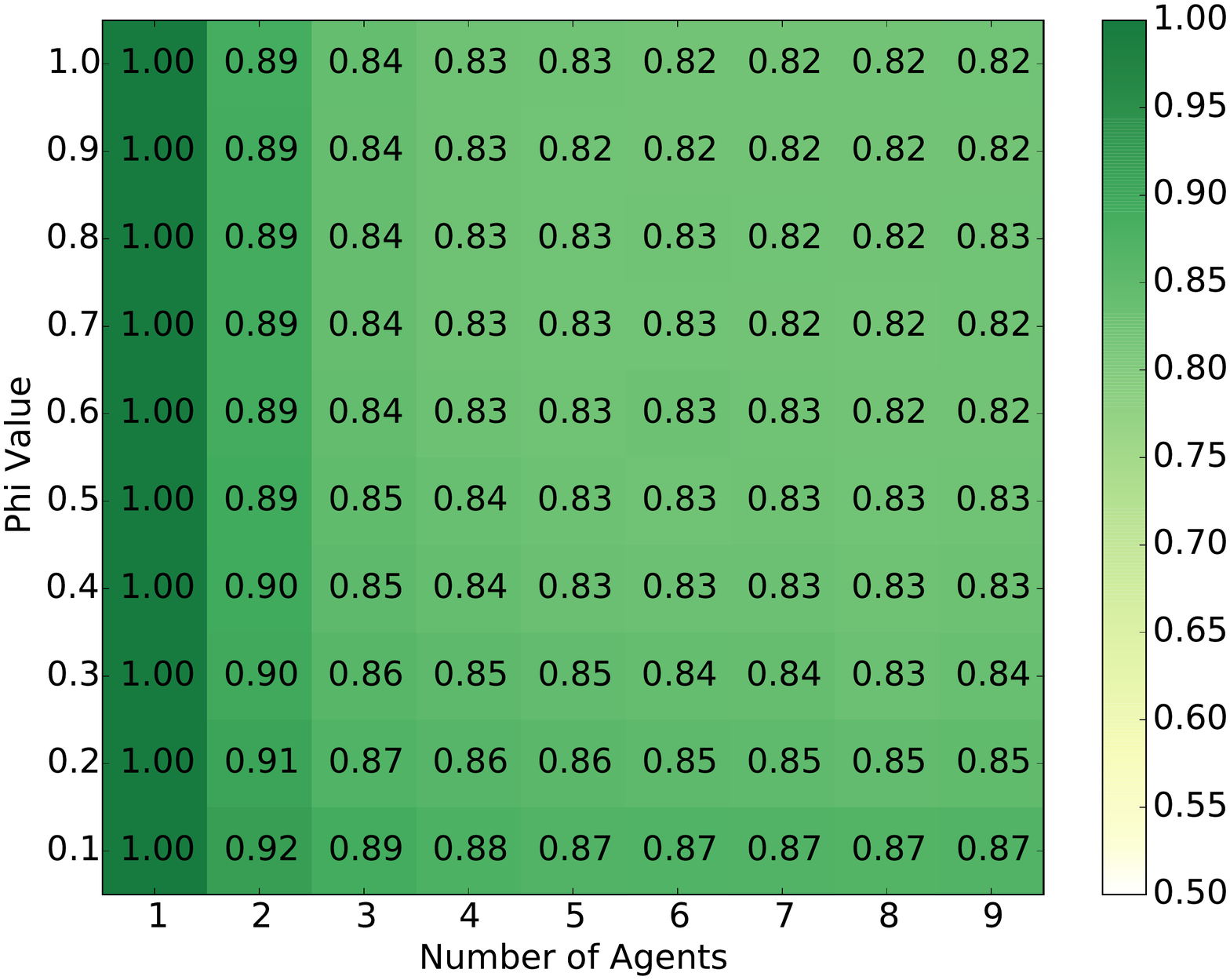}
		\endminipage
	
	\caption{Minimum (top) and average (bottom) achieved approximation ratio for the RSD (left) and PS (right) 
	mechanisms with Borda utilities.
	Observe that both mechanisms preform similarly and significantly better than the derived $\guar(J)$.  Both
	mechanisms are relatively invariant to the level of dispersion in the underlying valuation profiles.
	For each $n=\{1, \ldots, 9\}$ the graphs are aggregated
	over the complete range of objects.  For example, the cell $(n=4, \phi=0.2)$ is the minimum (resp. average) 
	achieved approximation ratio over all instances with $m \in \{1, \ldots, 9\}$.}
	\label{fig:min-avg-borda}
	\end{figure}

When $\phi=0.0$ (not shown in our graphs), the achieved approximation ratio is $1$ for both PS and RSD.  Both of these mechanisms return the uniform allocation, assigning probability of $\nicefrac{1}{n}$ for each object to each agent, when all the agents have identical preferences.  In general, sweeping the value of $\phi$ from completely correlated to completely uncorrelated preferences has little impact on the overall achievable approximation ratio, though for both models the achievable approximation ratio did strictly decrease as we increased $\phi$.  Comparing Figures~\ref{fig:min-avg-borda} and \ref{fig:min-avg-exp}, the impact of changing $\phi$ was strictly greater for the exponential utility model than it was for the Borda utility model, highlighting again that, as the difference between the valuations of the objects grows large, it becomes harder to achieve fair allocations.

\begin{figure}[ht!]
		\minipage[t][][b]{0.45\textwidth}
		\centering
		 \footnotesize{RSD Min. Achieved Approx. Ratio \\ Exponential Utilities}
		  \includegraphics[scale=0.20]{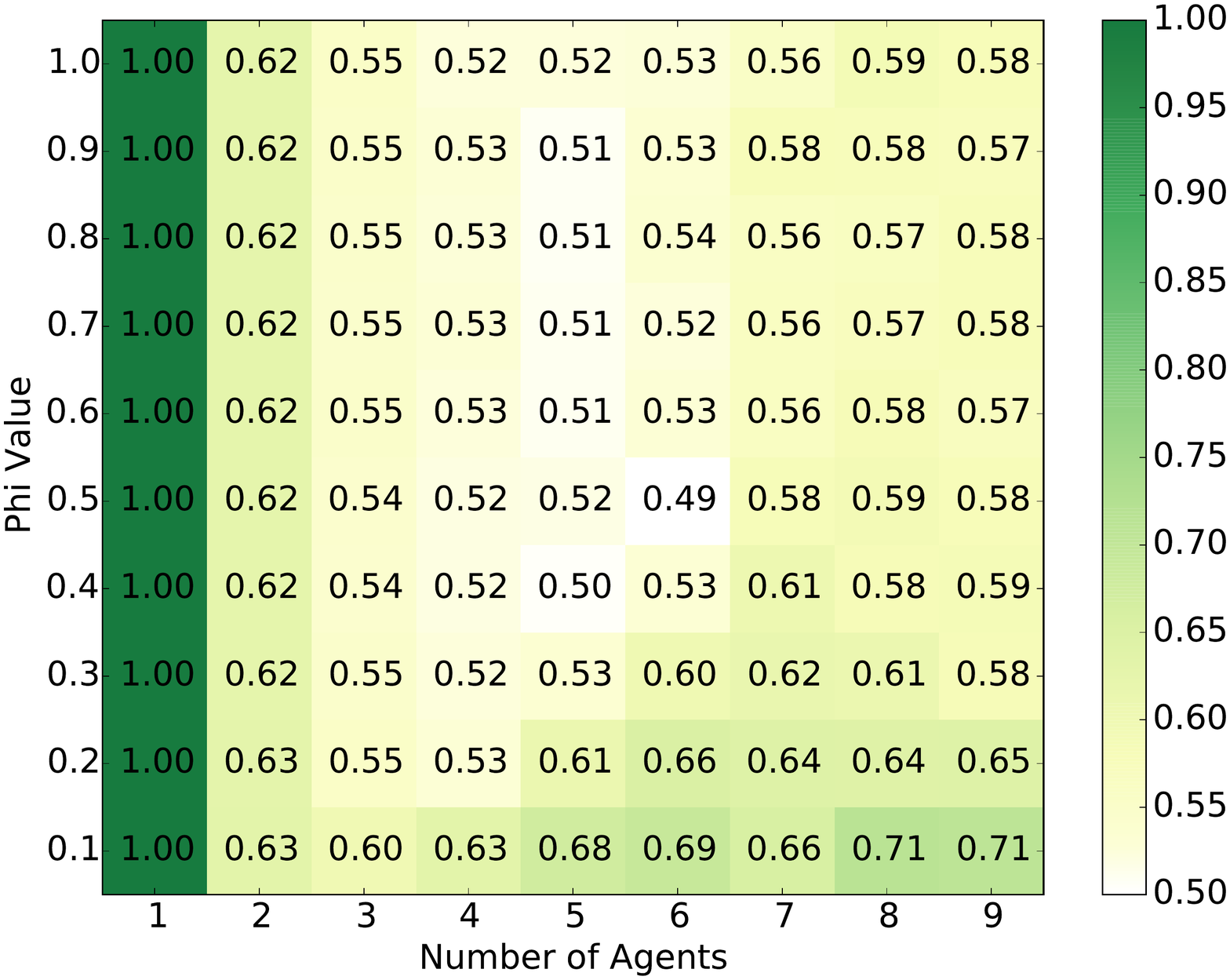}
		\endminipage
	\hfill
		\minipage[t][][b]{0.45\textwidth}
		\centering
		\footnotesize{PS Min. Achieved Approx. Ratio \\ Exponential Utilities}
  		\includegraphics[scale=0.20]{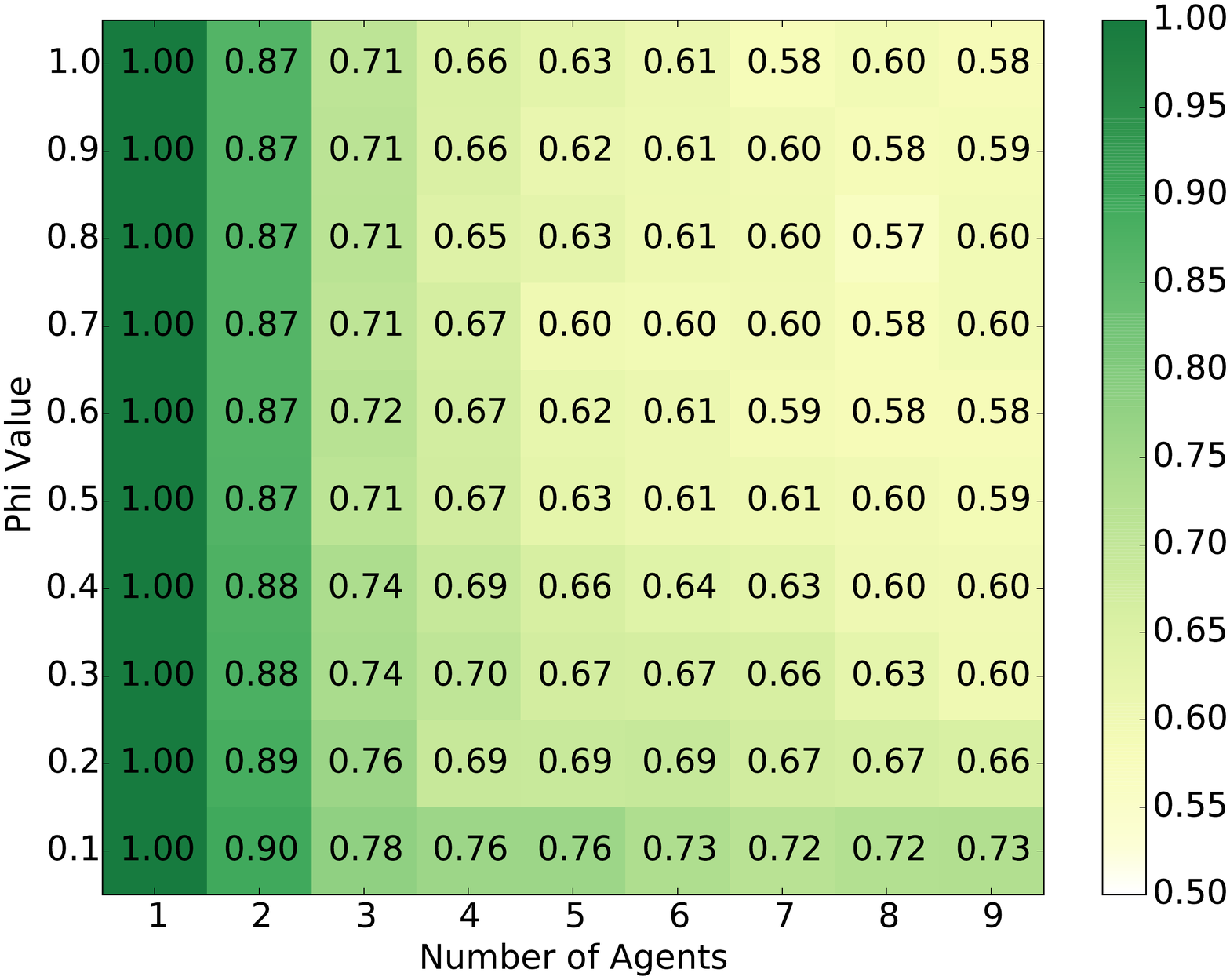}
		\endminipage
		
		\vspace{0.5cm}

		\minipage[t][][b]{0.45\textwidth}
		\centering
		\footnotesize{RSD Mean Achieved Approx. Ratio \\ Exponential Utilities}
	 	\includegraphics[scale=0.20]{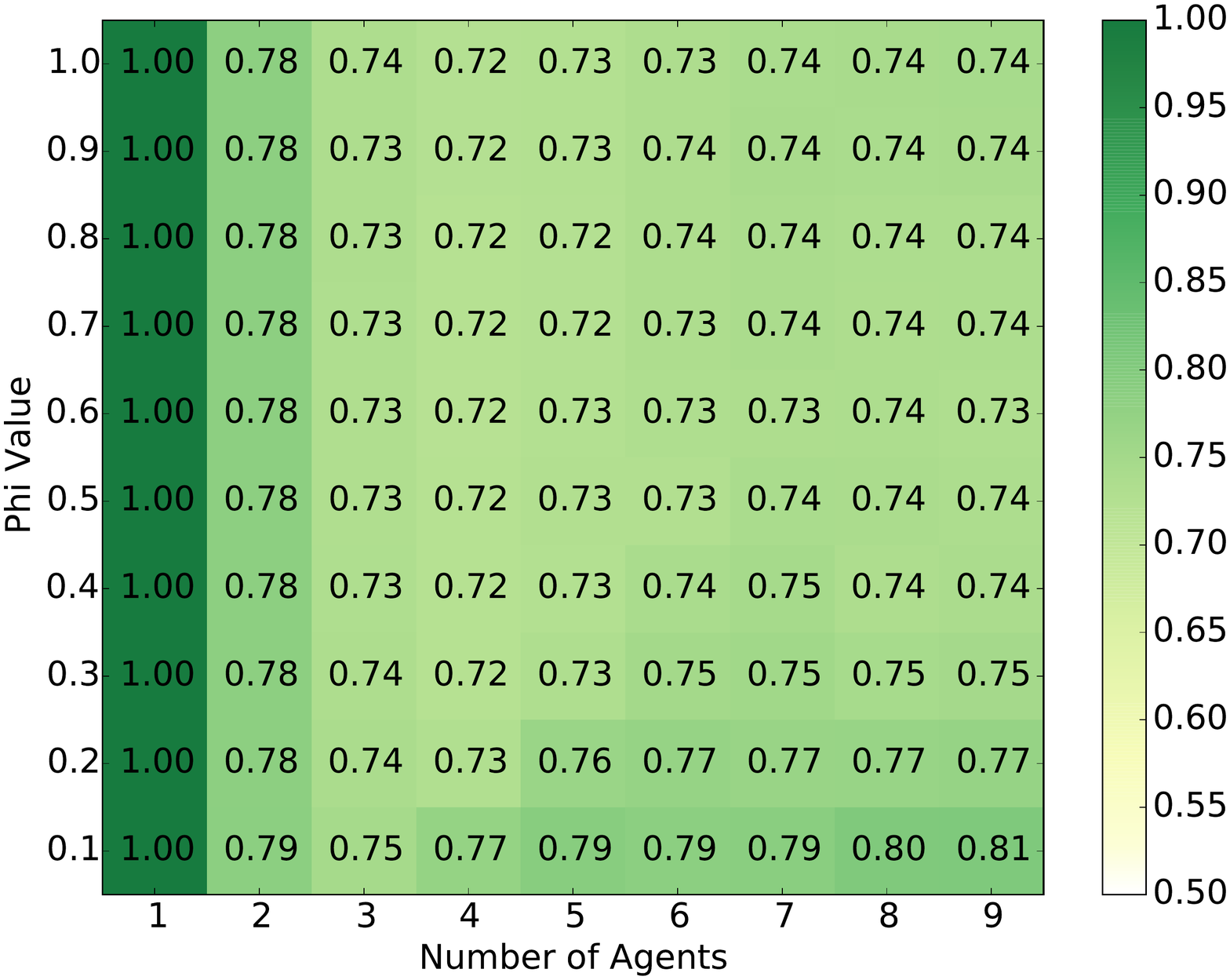}
		\endminipage
	\hfill
		\minipage[t][][b]{0.45\textwidth}
		\centering
		\footnotesize{PS Mean Achieved Approx. Ratio \\ Exponential Utilities}
  		\includegraphics[scale=0.20]{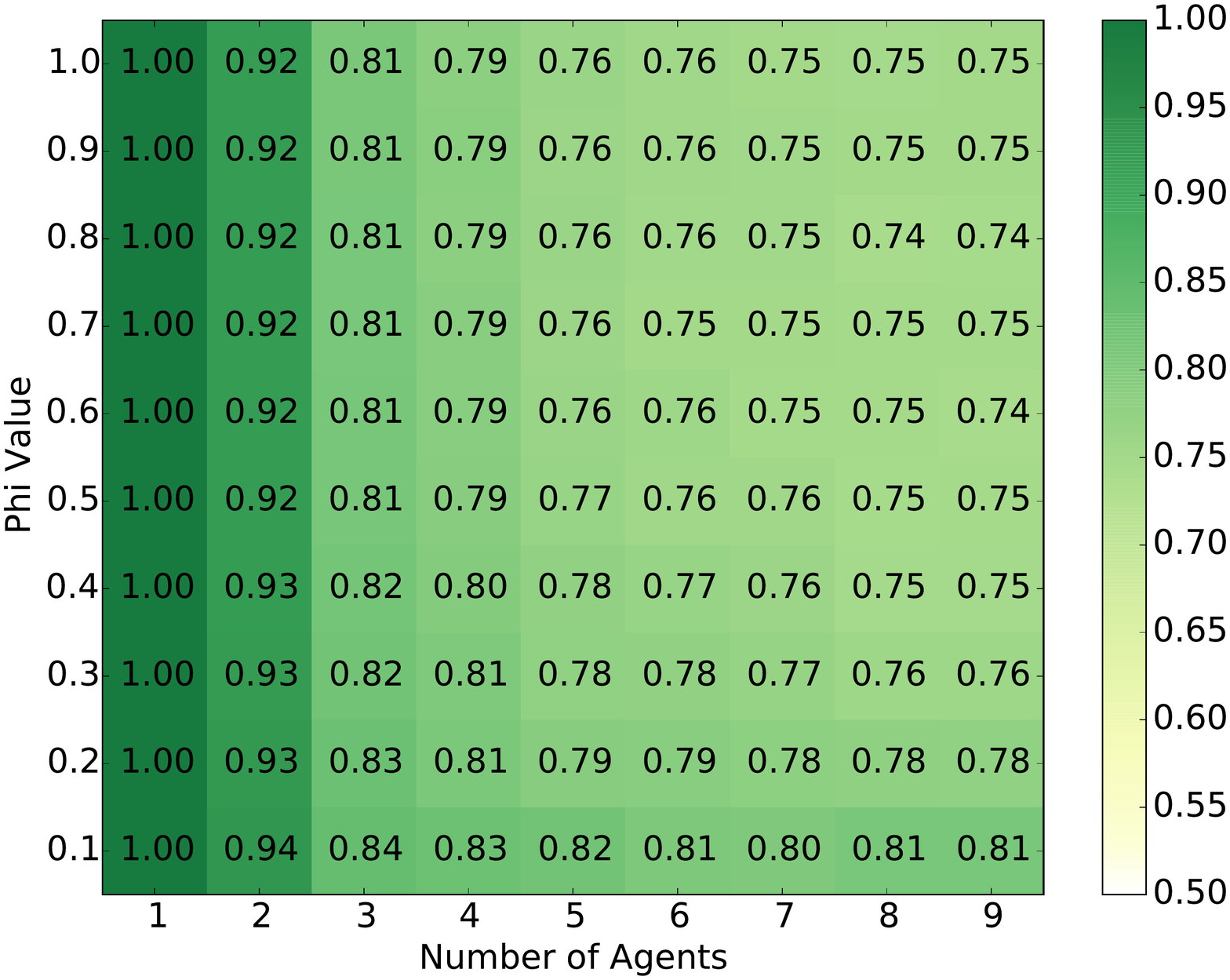}
		\endminipage
	
	\caption{Minimum (top) and average (bottom) achieved approximation ratio for the for the 
	RSD (left) and PS (right) mechanism with exponential utilities.
	These results are strictly dominated by those in Figure~\ref{fig:min-avg-borda}; implying that as 
	difference between the valuations of the objects grows the achieved approximation ratio decreases.	
	For each $n=\{1, \ldots, 9\}$ these graphs are aggregated
	over the complete range of objects.  For example, the cell $(n=4, \phi=0.2)$ is the minimum (resp. average) 
	approximation ratio achieved over all instances with $m \in \{1, \ldots, 9\}$.}
	\label{fig:min-avg-exp}
\end{figure}	

Interestingly, in Figure~\ref{fig:min-avg-borda} there appears to be almost no difference between the minimum and average ratios for PS and RSD under Borda utilities.  Furthermore, these ratios appear to be very high compared to our theoretical results.  %This gives us hope that the worst case results we saw in Section~\ref{sec:theory} will not often occur in practice.  
Finally, PS consistently performs slightly better than RSD for the minimum and average ratios under exponential utilities and on par with RSD for Borda utilities.  This provides more empirical support to the argument that PS is superior to RSD in terms of fairness.
\vspace{-1em}	
\subsection{Experiments: The Effect of Envy-freeness}

In order to evaluate the effect that envy-freeness has on the allocations we turn to the OEEF mechanism. To understand the worst case effects of adding envy-freeness as a hard constraint has on small instances we exhaustively tested the parameter space with agents $n \in \{2, \ldots, 6\}$, number of objects $m \in \{3, \ldots, 4\}$ under Borda and exponential utilities.  In this entire parameter space, the worse case achievable approximation ratio was $0.87$, significantly higher than the theoretical worst case.  This shows that for smaller instances and some standard utility models, the requirement of envy-freeness does not have a significant negative impact on the achievable approximation ratio.
To get an understanding of the performance of OEEF in a larger parameter space we repeated the experiments from the previous section here, evaluating the performance of OEEF across a large parameter space with number of agents $n \in \{2, \ldots, 9\}$, number of objects $m \in \{2, \ldots, 9\}$, and dispersion parameter $\phi \in \{0.0, 0.1, \ldots, 1.0\}$.  The results of these tests, again aggregated by $m$ and $\phi$, are shown in Figure~\ref{fig:envy-free}.

% \vspace{-2em}
\begin{figure}[ht]
	\minipage[t][][b]{0.45\textwidth}
	\centering
	\footnotesize{Min. Achieved Approx. Ratio \\ Envy-free and Borda Utilities}
	\includegraphics[scale=0.20]{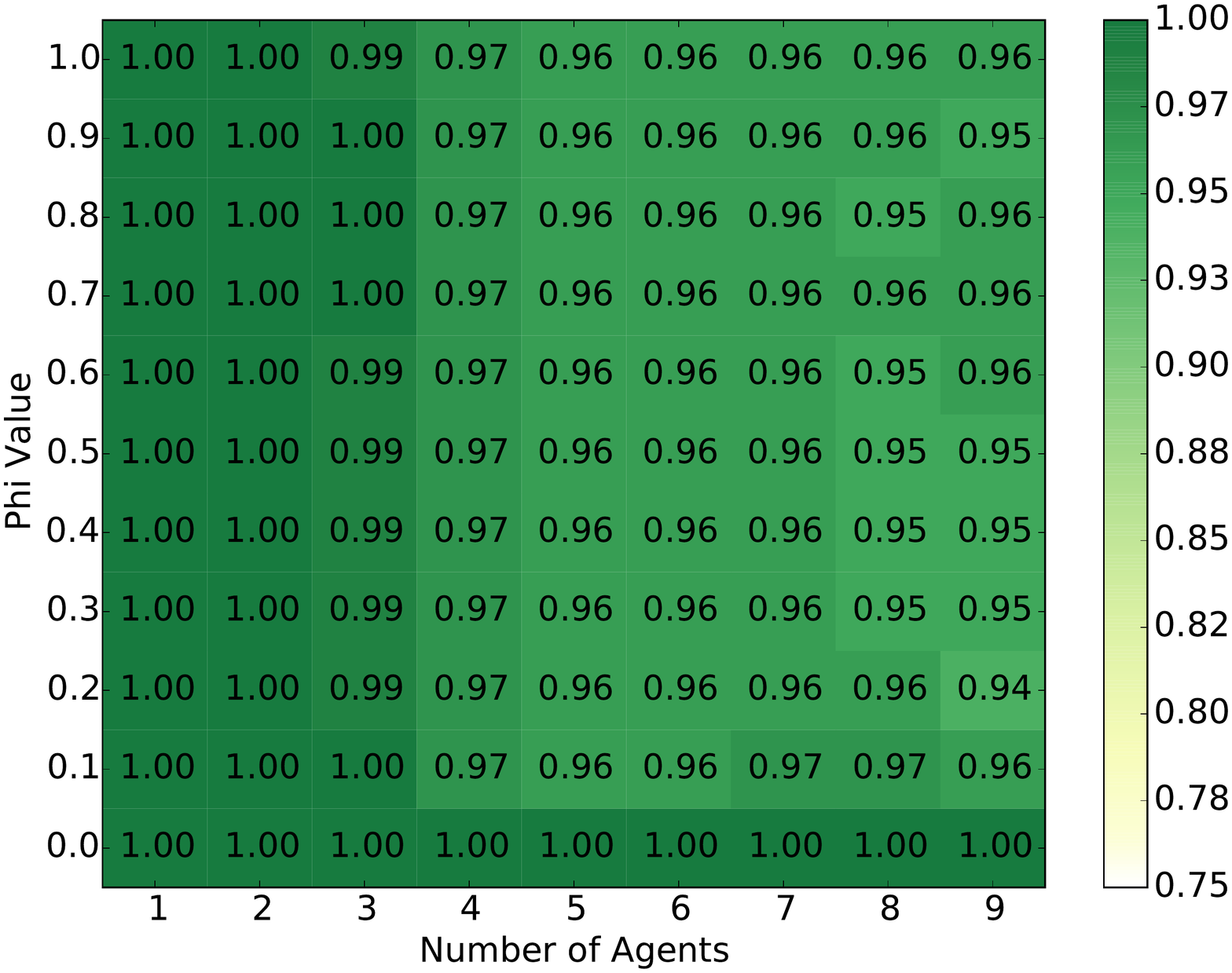}
	\endminipage 
	\hfill
	\minipage[t][][b]{0.45\textwidth}
	\centering
	\footnotesize{Min. Achieved Approx. Ratio \\ Envy-free and Exponential Utilities}
	\includegraphics[scale=0.20]{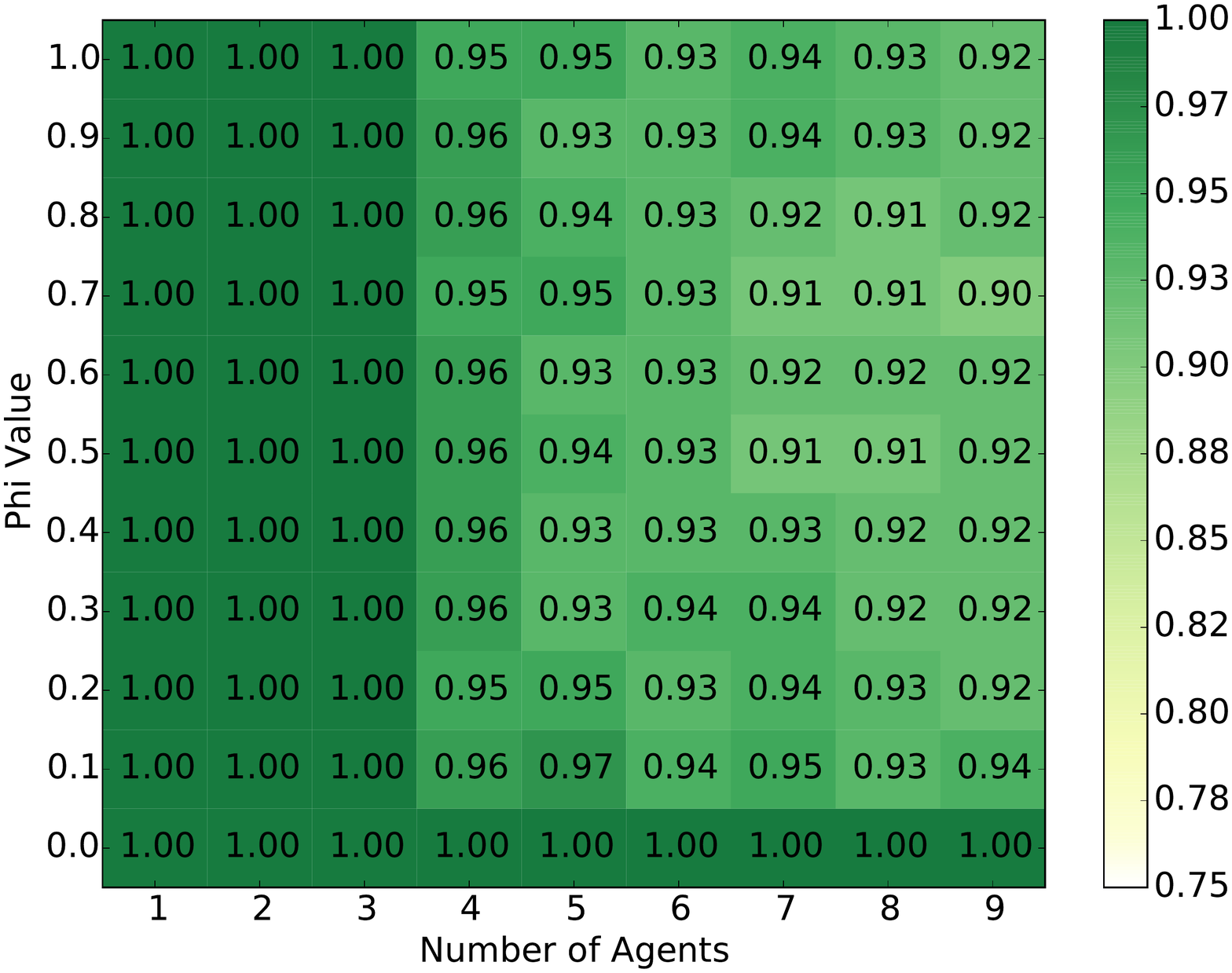}
	\endminipage 
	\caption{Minimum achieved approximation ratio when we enforce envy-freeness as a hard constraint.
	Contrasting these results with those from Figures~\ref{fig:min-avg-borda} and \ref{fig:min-avg-exp} for the PS mechanism demonstrates 
	that the stronger notion of envy-freeness that is satisfied by PS has a significant impact on the achieved
	approximation ratio while the more common notion of envy-freeness does not.
	For each $n=\{1, \ldots, 9\}$ these graphs are aggregated
	over the complete range of objects.  For example, the cell $(n=4, \phi=0.2)$ is the minimum (resp. average) 
	approximation ratio achieved over all instances with $m \in \{1, \ldots, 9\}$.}
	\label{fig:envy-free}
\end{figure}
% \vspace{-2em}

	%
	% We then generated ordinal preferences using Mallow's model for different values of $\phi$ and then obtained cardinal utilities via the Borda and exponential scoring functions.
When agents have exponential utilities, the achieved approximation ratio, much like in the last section, is strictly worse.  Additionally, when we have exponential utilities, as $\phi$ increases, the approximation ratio for the envy-free mechanisms first decreases slightly and then increases for higher number of agents.  Since $\phi=1.0$ means that agents preferences are drawn uniformly at random, it is more likely that each agent has high valuation for different objects.  Hence, as the preferences move from concentrated to dispersed, there seems to be an interesting transition from high to low and back to high in terms of the achievable approximation ratio.
As in the previous subsection, we observe that when all agents have the same preferences, the uniform allocation is both envy-free and has maximal achieved approximation ratio.  Hence, when $\phi=0$, the ratio is $1.0$ (not shown in Figure~\ref{fig:envy-free}).  We note that RSD preforms much more poorly across the board compared to OEEF.  
%This makes Conjecture~\ref{conj:strategyproof} an interesting open problem: our experimental results indicate that we must pay a higher price to ensure strategyproofness than we do for envy-freeness \nick{I'm reaching here... feel free to remove the last sentence}.
	%The PS mechanism satisfies SD envy-freeness --- a more stringent notion of envy-freeness than the one required in our linear program.  
The results in Figure~\ref{fig:envy-free} strictly dominate the results for PS in Figures~\ref{fig:min-avg-borda} and \ref{fig:min-avg-exp}.  Hence, SD envy-freeness that is satisfied by PS has a significant impact on the achieved approximation ratio. 
% \nick{We need something here like: this provides support for continued investigation into the OEEF mechanism as it may give better results, in practice, than PS}.
	
%	We note that envy-freeness in itself does not have a big impact on the OEV. In contrast to PS that returns an envy-free allocation for all cardinal utilities consistent with the ordinal preferences has a bigger impact on the approximation of the OEV achieved. 
%	\nick{ Haris says: [that returns an envy- free allocated for all cardinal utilities consistent with the ordinal preferences]. Here we want to highlight the contrast of these with the previous two graphs for PS.  This underlines the fact that PS is ensuring a stronger notion of envy-freeness than the more standard defintion.}
%	\nick{Should we formally define this ratio again so we can see clearly in Fig 3?}
%	

\vspace{-1em}
\section{Conclusion}
We present theoretical and experimental results concerning how well different randomized mechanisms approximate the optimal egalitarian value. It has been well-known that egalitarianism can be incompatible with envy-free or truthfulness. In this paper, we quantified how much egalitarianism is affected by such properties. 
In a recent paper, \citet{CFK+15a} proved results for the utilitarian welfare of the Nash equilibria of assignment mechanisms. 
It will be interesting to adopt a similar approach with respect to the egalitarian value and study the price of anarchy of randomized mechanisms with respect to that objective.
%\haris{Commented this out since already in lit review}
%In their setting, agents strategise within a non-truthful mechanism and the outcome is evaluated in the Nash equilibria of the induced game, whereas in our investigations, when considering non-truthful mechanisms, agents are not strategic and the loss in the egalitarian value is purely due to other required properties, such as envy-freeness or ordinality. It will be interesting to adopt a similar approach with respect to the egalitarian value and study the price of anarchy of randomized mechanisms with respect to that objective. 
We assume \emph{additive} cardinal utilities in this paper, it would be interesting to consider other assumptions.
To conclude, we mention an open problem: what is the \emph{best} OEV approximation guaranteed by truthful mechanisms? 
%Finally, the OEEF mechanism may be of independent interest and perhaps deserves further attention. 

%Our worst case approximation results are for very particular and extreme utility functions. It will be interesting to obtain more positive guarantees for restricted utility functions such as Borda scores.  
%If we consider 0/1 utilities, then a variant of PS that works for indifferent already achieves OEV~\citep{BoMo04a,KaSe06a}. We mention an open problem: what is the best OEV approximation guaranteed by truthful mechanisms? Finally, the OEEF mechanism may be of independent interest and perhaps deserves further attention. 

\small
\vspace{-1em}
%\clearpage
\section*{Acknowledgments}

NICTA is funded by the Australian Government through the Department of Communications and the Australian Research Council through the ICT Centre of Excellence Program. Aris Filos-Ratsikas was supported by the Sino-Danish Center for the Theory of Interactive Computation, funded by the Danish National Research Foundation and the National Science Foundation of China (under the grant 61061130540), and by the Center for research in the Foundations of Electronic Markets (CFEM), supported by the Danish Strategic Research Council.

 \vspace{-1em}
   %\bibliography{../../../research/papers/pamas/abb,../../../research/papers/pamas/pamas,../../../research/papers/pamas/aziz,../../../research/papers/pamas/brandt}

%   \bibliography{../../../research/papers/pamas/abbshort,../../../research/papers/pamas/pamas,../../../research/papers/pamas/aziz,../../../research/papers/pamas/brandt,../../adtbib/abb,../../adtbib/adt}
% %

\end{document}